\documentclass[11pt]{amsart}
\def\isdraft{0}

\usepackage[dvipsnames]{xcolor}
\usepackage{bbm}
\usepackage{graphicx}
\usepackage{amsaddr} 
\usepackage{mathtools}
\usepackage{amsmath,amssymb,amsfonts,dsfont}
\usepackage{prettyref} 
\usepackage[utf8]{inputenc}
\usepackage{enumitem}
\usepackage[hyphens]{url} 
\usepackage{tikz-cd}
\usepackage{tikz}
\usetikzlibrary{positioning}
\usetikzlibrary{arrows}
\usepackage{bussproofs}
\usepackage[mathscr]{euscript}
\usepackage{hyperref}
\usepackage{amsthm}
\usepackage{a4wide}
\usepackage[color=black,textcolor=white\if\isdraft0,disable\fi]{todonotes}
\usepackage{stackengine}
\usepackage{stmaryrd}
\usepackage{wrapfig}
\usepackage{marginnote}
\usepackage{float}
\graphicspath{{Figures/}}
\tikzset{every state/.style={minimum size=0pt}}
\usepackage{footnote}
\usepackage{tabularx}
\usepackage{txfonts}
\usepackage[normalem]{ulem}
\usepackage{pifont}
\usepackage[most]{tcolorbox}
\usepackage{shuffle}

\newtheorem{theorem}{Theorem}

\newtheorem{corollary}[theorem]{Corollary}

\newtheorem{lemma}[theorem]{Lemma}
\newtheorem{proposition}[theorem]{Proposition}

\theoremstyle{definition} 

\newtheorem{definition}[theorem]{Definition}

\newtheorem{example}[theorem]{Example}

\newtheorem{remark}[theorem]{Remark}

\newrefformat{§}{§\ref{#1}}
\newrefformat{a}{Answer \ref{#1}}
\newrefformat{ana}{Analogy \ref{#1}}
\newrefformat{alg}{Algorithm \ref{#1}}
\newrefformat{ax}{Appendix \ref{#1}}
\newrefformat{cl}{Claim \ref{#1}}
\newrefformat{con}{Conclusion \ref{#1}}
\newrefformat{conv}{Convention \ref{#1}}
\newrefformat{cj}{Conjecture \ref{#1}}
\newrefformat{c}{Corollary \ref{#1}}
\newrefformat{q}{(\ref{#1})}
\newrefformat{e}{Example \ref{#1}}
\newrefformat{exe}{Exercise \ref{#1}}
\newrefformat{d}{Definition \ref{#1}}
\newrefformat{Done}{Done \ref{#1}}
\newrefformat{f}{Fact \ref{#1}}
\newrefformat{fig}{Figure \ref{#1}}
\newrefformat{h}{Hypothesis \ref{#1}}
\newrefformat{idea}{Idea \ref{#1}}
\newrefformat{i}{Item \ref{#1}}
\newrefformat{l}{Lemma \ref{#1}}
\newrefformat{note}{Note \ref{#1}}
\newrefformat{n}{Notation \ref{#1}}
\newrefformat{o}{Observation \ref{#1}}
\newrefformat{pr}{Problem \ref{#1}}
\newrefformat{p}{Proposition \ref{#1}}
\newrefformat{pc}{Pseudocode \ref{#1}}
\newrefformat{Q}{Q. \ref{#1}}
\newrefformat{r}{Remark \ref{#1}}
\newrefformat{ta}{Table \ref{#1}}
\newrefformat{t}{Theorem \ref{#1}}
\newrefformat{T}{Todo \ref{#1}}
\newrefformat{w}{Warning \ref{#1}}

\newcommand{\boxblacktriangle}{\mathrel{\ooalign{$\square$\cr\kern0.07ex\hbox{\scalebox{0.9}{$\blacktriangle$}}}}}

\newcommand{\boxtriangle}{\mathrel{\ooalign{$\square$\cr\kern0.07ex\hbox{\scalebox{0.9}{$\triangle$}}}}}

\newcommand{\righttherefore}{:\joinrel\cdot\,}

\usepackage[top=0.7cm,bottom=0.5cm,left=1cm,right=1cm]{geometry} 

\if\isdraft1
\fi

\title{
	Logic-based analogical proportions
}
\author{
	Christian Anti\'c
}
\address{
	christian.antic@icloud.com\\
	Vienna University of Technology\\
	Vienna, Austria
}

\begin{document}
\begin{abstract} 
	The author has recently introduced an abstract algebraic framework of analogical proportions within the general setting of universal algebra. The purpose of this paper is to lift that framework from universal algebra to the strictly more expressive setting of full first-order logic. We show that the so-obtained logic-based framework preserves all desired properties and we prove novel results in that extended setting.
	\todo[inline]{Erwaehne, dass das ein signifikanter Fortschritt zu state of the art ist, da erstes FW in FOL}
\end{abstract}
\maketitle

\section{Introduction}\label{§:I}




The author has recently introduced an abstract algebraic justification-based framework of analogical proportions of the form ``$a$ is to $b$ what $c$ is to $d$'' --- written $a:b::c:d$ --- in the general setting of universal algebra \cite{Antic22}. It has been applied to logic program synthesis by analogy \cite{Antic23-23} (which can be intrepreted as a form of analogical proportions between first-order Horn theories \cite{Herzig24}), and it has been studied in the context of boolean \cite{Antic21-3} and monounary algebras \cite{Antic22-2}.

The \textbf{purpose of this paper} is to lift that model from universal algebra to full first-order logic motivated by the fact that some reasoning tasks necessarily involve quantifiers and relations. The entry point is the \textit{logical} interpretation of an analogical proportion in \cite[§6]{Antic22}. The task of turning that limited logical interpretation --- restricted to so-called rewrite formulas representing rule-like justifications --- into a full-fledged logical description of analogical proportions turns out to be non-trivial due to the observation that a naive extension easily leads to an over-generalization with too many elements being in proportion: for example, consider the structure $(\mathbb N,S,0)$, where $S$ is the successor function; in this structure, we can identify every natural number $a$ with the numeral $\underline a:=S^a0$; given some natural numbers $a,b,c,d\in \mathbb N$, the formula
\begin{align*} 
	\alpha(x,y):\equiv (x=\underline a\land y=\underline b)\lor (x=\underline c\land y=\underline d)
\end{align*} would be a characteristic justification (see \prettyref{§:AP}) of
\begin{align*} 
	a:b::c:d
\end{align*} since
\begin{align*} 
	(\mathbb N,S,0) \models \alpha(a',b') \quad\text{and}\quad (\mathbb N,S,0) \models \alpha(c',d')
\end{align*} holds iff
\begin{align*} 
	(a'=a,\quad b'=b,\quad c'=c,\quad d'=d) \quad\text{or}\quad (a'=c,\quad b'=d,\quad c'=a,\quad d'=b).
\end{align*}

The \textbf{main challenge} therefore is to find an appropriate fragment of first-order logic expressing the relationship between two given elements. This is achieved in this paper via the notion of ``connected formulas'' (see \prettyref{§:AP}). From that point on, the paper is similar to \cite{Antic22} in spirit but different on a technical level due to the more expressive justifications possibly containing relation symbols and conjunctions of atoms. 

In \prettyref{§:Properties}, we show that the extended framework of this paper preserves all the desirable properties proved in \cite[Theorem 28]{Antic22} (based on Lepage's \cite{Lepage03} axiomatic approach) and thus coincides in that respect with the original framework. After that, some results are lifted to the new setting, most importantly the Isomorphism Theorems of \prettyref{§:ITs}. 

In \prettyref{§:EPT}, we present an interesting new result linking \textit{equational dependencies} and analogical proportions. It should be emphasized that these kind of results cannot be shown in the previous version of the framework due to the inability of explicitly representing equations via rewrite justifications. 

In \prettyref{§:EF}, we reprove the Difference Proportion Theorem in \cite{Antic22-2}, stating that in the structure $(\mathbb N,S)$ consisting of the natural numbers and the unary successor function we have
\begin{align*} 
	a:b::c:d \quad\Leftrightarrow\quad a-b=c-d,
\end{align*} in a fragment consisting only of justifications of an equational form.

In \prettyref{§:Graphs}, we study graphs within the path fragment consisting only of path justifications of a specific form encoding path lengths.


In a broader sense, this paper is a further step towards an algebro-logical theory of analogical reasoning.

\section{Preliminaries}

We recall the syntax and semantics of first-order logic by mainly following the lines of \cite[§2]{Hinman05}.

\subsection{Syntax}

A (\textit{\textbf{first-order}}) \textit{\textbf{language}} $L$ consists of a set $Rs_L$ of \textit{\textbf{$L$-relational symbols}}, a set $Fs_L$ of \textit{\textbf{$L$-function symbols}}, a set $Cs_L$ of \textit{\textbf{$L$-constant symbols}}, and a function $r:Fs_L\cup Rs_L\to \mathbb N$. The sets $Rs_L$, $Fs_L$, and $Cs_L$ are pairwise disjoint, and members of $Rs_L\cup Fs_L\cup Cs_L$ are called the \textit{\textbf{non-logical symbols}} of $L$.  Additionally, every language has the following distinct \textit{\textbf{logical symbols}}: a denumerable set $X$ of \textit{\textbf{variables}}, the \textit{\textbf{equality symbol}} $=$, the \textit{\textbf{connectives}} $\neg$, $\lor$, and $\land$, and the \textit{\textbf{quantifiers}} $\exists$ and $\forall$.

An \textit{\textbf{$L$-atomic term}} is either a constant symbol or a variable of $L$. An \textit{\textbf{$L$-term}} is defined inductively as follows:
\begin{itemize}
	\item every $L$-atomic term is an $L$-term,
	\item for any function symbol $f$ and any $L$-terms $t_1,\ldots,t_{r(f)}$, $f(t_1,\ldots,t_{r(f)})$ is an $L$-term.
\end{itemize} We denote the set of variables occurring in a term $t$ by $X(t)$.\todo{needed?} The \textit{\textbf{rank}} of a term is given by the number of its variables.

An \textit{\textbf{$L$-atomic formula}} has one of the following forms:
\begin{itemize}
	\item $s=t$, for $L$-terms $s,t$;
	\item $p(t_1,\ldots,t_{r(p)})$, for an $L$-relational symbol $p$ and $L$-terms $t_1,\ldots,t_{r(p)}$.
\end{itemize} An \textit{\textbf{$L$-formula}} is defined inductively as follows:
\begin{itemize}
	\item every $L$-atomic formula is an $L$-formula;
	\item if $\alpha$ and $\beta$ are $L$-formulas, then so are $\neg\alpha$, $\alpha\lor\beta$, and $\alpha\land\beta$;
	\item if $\alpha$ is an $L$-formula and $x\in X$ is a variable, then $(\exists x)\alpha$ and $(\forall x)\alpha$ are $L$-formulas.
\end{itemize} The \textit{\textbf{rank}} of an $L$-formula is the number of its free variables, where a variable is called \textit{\textbf{free}} iff it is not in the scope of a quantifier. We denote the set of variables occurring in $\alpha$ (not necessarily free) by $X\alpha$. We expect that quantified variables are distinct which means that we disallow formulas of the form $(\forall x)Px\land (\exists x)Rx$.

\subsection{Semantics}

An \textit{\textbf{$L$-structure}} is specified by a non-empty set $A$, the \textit{\textbf{universe}} of $\mathfrak A$; for each $p\in Rs_L$, a relation $p^ \mathfrak A\subseteq A^{r(p)}$, the \textit{\textbf{relations}} of $\mathfrak A$; for each $f\in Fs_L$, a function $f^ \mathfrak A:A^{r(f)}\to A$, the \textit{\textbf{functions}} of $\mathfrak A$; for each $c\in Cs_L$, an element $c^ \mathfrak A\in A$, the \textit{\textbf{distinguished elements}} of $\mathfrak A$.

Every term $s$ induces a function $s^ \mathfrak A:A^{r(s)}\to A$ in the usual way.

We define the \textit{\textbf{logical entailment relation}} inductively as follows: for any $L$-structure $\mathfrak A$, $L$-terms $s,t$, $L$-formulas $\alpha,\beta$, and $ \textbf{a}\in A^{r(\alpha)-1}$,
\begin{align*} 
	\mathfrak A \models s=t \quad&:\Leftrightarrow\quad s^ \mathfrak A=t^ \mathfrak A,\\
	\mathfrak A \models p(t_1\ldots t_{r(p)}) \quad&:\Leftrightarrow\quad (t_1^ \mathfrak A\ldots t_{r(p)}^ \mathfrak A)\in p^ \mathfrak A,\\ 
	\mathfrak A \models\neg\alpha \quad&:\Leftrightarrow\quad \mathfrak A\not\models \alpha,\\
	\mathfrak A\models \alpha\lor\beta \quad&:\Leftrightarrow\quad \mathfrak A\models \alpha \quad\text{or}\quad \mathfrak A\models\beta,\\
	\mathfrak A\models \alpha\land\beta \quad&:\Leftrightarrow\quad \mathfrak A\models \alpha \quad\text{and}\quad \mathfrak A\models \beta,\\
	\mathfrak A\models (\exists x)\alpha(\textbf{a},x) \quad&:\Leftrightarrow\quad \mathfrak A\models \alpha(\textbf{a},b),\quad\text{for some $b\in A$},\\
	\mathfrak A\models (\forall x)\alpha(\textbf{a},x) \quad&:\Leftrightarrow\quad \mathfrak A\models \alpha(\textbf{a},b),\quad\text{for all $b\in A$}.
\end{align*} 

A \textit{\textbf{homomorphism}} from $\mathfrak A$ to $\mathfrak B$ is a mapping $H : A\to B$ such that for any function symbol $f$ and any sequence of elements $a_1,\ldots,a_{r(f)}\in A^{r(f)}$,
\begin{align*} 
	H(f^\mathfrak A(a_1,\ldots,a_{r(f)})) = f^\mathfrak B(H(a_1),\ldots,H(a_{r(f)})),
\end{align*} An \textit{\textbf{isomorphism}} is any bijective homomorphism.

Let $\mathfrak A$ and $\mathfrak B$ be $L$-structures. We say that a mapping $F: \mathfrak A\to \mathfrak B$ \textit{\textbf{respects}}
\begin{itemize}
	\item a term $t$ iff for each $\textbf{a}\in A^{r(t)}$,
	\begin{align*} 
		F(t^\mathfrak A( \textbf{a}))=t^\mathfrak B(F( \textbf{a})).
	\end{align*}

	\item a formula $\alpha$ iff for each $\textbf{a}\in A^{r(\alpha)}$,
	\begin{align*} 
		\mathfrak A\models \alpha( \textbf{a}) \quad\Leftrightarrow\quad \mathfrak B\models \alpha(F( \textbf{a})).
	\end{align*}
\end{itemize}

The following result will be useful in \prettyref{§:ITs} for proving our First Isomorphism \prettyref{t:FIT}; its straightforward induction proof can be found, for example, in \cite[Lemma 2.3.6]{Hinman05}:

\begin{lemma}\label{l:respects} Isomorphisms respect $L$-terms and formulas.
\end{lemma}

\section{Analogical proportions}\label{§:AP}

In this section, we lift the algebraic framework of analogical proportions in \cite{Antic22} from universal algebra to first-order logic. In what follows, let $L$ be a first-order language and let $\mathfrak A$ and $\mathfrak B$ be $L$-structures. 

The entry point is the logical interpretation of analogical proportions in terms of model-theoretic types in \cite[§6]{Antic22}. The main difficulty is to find an appropriate fragment of first-order formulas expressing the relationship between two given elements so that all relevant properties can be expressed without including inappropriate ones which may easily lead to an over-generalization putting too many elements in proportion (see the discussion in \prettyref{§:I}). This is achieved in this paper by introducing the notion of a connected formula:

\begin{definition} A \textit{\textbf{2-$L$-formula}} is a formula containing exactly two free variables $x$ and $y$. The set of \textit{\textbf{conjunctive $L$-formulas}} consists of $L$-formulas not containing negation\todo{allow?} or disjunction.

We define the undirected \textit{\textbf{dependency graph}} of a conjunctive 2-$L$-formula $\alpha$ as follows:
\begin{itemize}
	\item The set of vertices is given by the set $X\alpha$ of all variables occurring in $\alpha$.\todo{include constant symbols?}
	\item There is an (undirected) edge $\{w,z\}$ between two variables $w,z\in X\alpha$ iff $w$ and $z$ both occur in an atomic formula in $\alpha$. 
\end{itemize} We call a conjunctive $L$-formula $\alpha$ a \textit{\textbf{connected $L$-formula}} (or \textit{\textbf{c-formula}}) iff the dependency graph of $\alpha$ is a connected graph, meaning that there is a path between any two vertices, containing both variables $x$ and $y$. We denote the set of all c-formulas over $L$ by $c\text-Fm_L$. A \textit{\textbf{c-term}} (resp., \textit{\textbf{c-atom}}) is an $L$-term (resp., $L$-atomic formula) containing both variables $x$ and $y$.
\todo[inline]{$x=a\land y=a$ is not connected --- do you want that?}
\end{definition}

\begin{example} The dependency graph of the formula
\begin{align*} 
	(\exists w)(\exists z)(x=y\land w=z)
\end{align*} is given by the disconnected graph
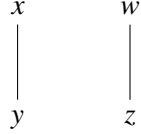
\begin{center}
\begin{tikzpicture} 
	\node (y_1) {$x$};
	\node (y_2) [below=of y_1] {$y$};
	\node (z_1) [right=of y_1] {$w$};
	\node (z_2) [below=of z_1] {$z$};
	\draw (y_1) to (y_2);
	\draw (z_1) to (z_2);
\end{tikzpicture}
\end{center} which means that the formula is not a connected formula. Roughly speaking, the subformula $w=z$ does not contain any information about the relationship between $x$ and $y$ and is therefore considered redundant. The reduced formula $x=y$, on the other hand, is easily seen to be connected.
\end{example}

The following definition --- which is an adaptation of a more restricted definition given in the setting of universal algebra \cite[Definition 8]{Antic22} --- is motivated by the observation that analogical proportions of the form $a:b::c:d$ are best defined in terms of arrow proportions $a\to b\righttherefore c\to d$ formalizing \textit{directed} relations and a maximality condition on the set of justifications. More precisely, to say that ``$a$ is related to $b$ as $c$ is related to $d$'' means that the set of justifications in the form of connected formulas $\alpha$ such that $\alpha(a,b)$ and $\alpha(c,d)$ is maximal with respect to $d$, which intuitively means that the relation $a\to b$ is maximally similar to the relation $c\to d$.

\begin{definition}\label{d:abcd} Let $a,b\in A$ and $c,d\in B$. We define the \textit{\textbf{analogical proportion relation}} as follows:
\begin{enumerate}
	\item Define the \textit{\textbf{connected $L$-type}} (or \textit{\textbf{c-type}}) of an \textit{\textbf{arrow}} $a\to b$ in $\mathfrak A$ by
	\begin{align*} 
		\uparrow_\mathfrak A(a\to b):=\left\{\alpha\in c\text-Fm_L \;\middle|\; \mathfrak A\models \alpha(a,b)\right\},
	\end{align*} extended to an \textit{\textbf{arrow proportion}} $a\to b\righttherefore c\to d$ --- read as ``$a$ transforms into $b$ as $c$ transforms into $d$'' --- in $\mathfrak{(A,B)}$ by
	\begin{align*} 
		\uparrow_{\mathfrak{(A,B)}}( a\to b\righttherefore c\to d)
			&:=\ \uparrow_\mathfrak A(a\to b)\ \cap \uparrow_\mathfrak B(c\to d)
	\end{align*} We call every c-formula in $\uparrow_{\mathfrak{(A,B)}}( a\to b\righttherefore c\to d)$ a \textit{\textbf{justification}} of $a\to b\righttherefore c\to d$ in $\mathfrak{(A,B)}$.

	\item A justification is \textit{\textbf{trivial}} in $\mathfrak{(A,B)}$ iff it justifies every arrow proportion in $\mathfrak{(A,B)}$ and we denote the set of all such trivial justifications by $\emptyset_ \mathfrak{(A,B)}$. Moreover, we say that a set of justifications $J$ is a \textit{\textbf{trivial set of justifications}} in $\mathfrak{(A,B)}$ iff every justification in $J$ is trivial. 

	\item We say that $a\to b \righttherefore c\to d$ \textit{\textbf{holds}} in $\mathfrak{(A,B)}$ --- in symbols,
	\begin{align*} 
		\mathfrak{(A,B)} \models a\to b \righttherefore c\to d,
	\end{align*} iff
	\begin{enumerate}
		\item either $\uparrow_\mathfrak A(a\to b)\ \cup \uparrow_\mathfrak B(c\to d)=\emptyset_ \mathfrak{(A,B)}$ consists only of trivial justifications, in which case there is neither a non-trivial relation between $a$ and $b$ in $\mathfrak A$ nor between $c$ and $d$ in $\mathfrak B$;

		\item or $\uparrow_{\mathfrak{(A,B)}}(a\to b\righttherefore c\to d)$ contains at least one non-trivial justification and is maximal with respect to subset inclusion among the sets $\uparrow_{\mathfrak{(A,B)}}(a\to b\righttherefore c\to d')$, $d'\in B$, that is, for any element $d'\in B$,\footnote{In what follows, we will usually omit trivial justifications from notation. So, for example, we will write $\uparrow_{\mathfrak{(A,B)}}(a\to b\righttherefore c\to d)=\emptyset$ instead of $\uparrow_{\mathfrak{(A,B)}}(a\to b\righttherefore c\to d)=\{\text{trivial justifications}\}$ in case $a\to b\righttherefore c\to d$ has only trivial justifications in $\mathfrak{(A,B)}$, et cetera. The empty set is always a trivial set of justifications. Every justification is meant to be non-trivial unless stated otherwise.}
		\begin{align*} 
		    \emptyset_ \mathfrak{(A,B)}\subsetneq\ \uparrow_{\mathfrak{(A,B)}}(a\to b\righttherefore c\to d)&\subseteq\ \uparrow_{\mathfrak{(A,B)}}(a\to b\righttherefore c\to d')
		\end{align*} implies
		\begin{align*} 
		    \emptyset_ \mathfrak{(A,B)}\subsetneq\ \uparrow_{\mathfrak{(A,B)}}(a\to b\righttherefore c\to d')\subseteq\ \uparrow_{\mathfrak{(A,B)}}(a\to b\righttherefore c\to d).
		\end{align*} 
	\end{enumerate} 
	\todo[inline]{Abbreviation ``$d$-maximal''}

	\item Finally, the \textit{\textbf{analogical proportion relation}} is defined by
	\begin{align*} 
		\mathfrak{(A,B)} \models a:b::c:d \quad:\Leftrightarrow\quad 
			& \mathfrak{(A,B)} \models a\to b\righttherefore c\to d \quad\text{and}\quad \mathfrak{(A,B)} \models b\to a\righttherefore d\to c \quad\text{and}\\
			& \mathfrak{(B,A)} \models c\to d\righttherefore a\to b \quad\text{and}\quad \mathfrak{(B,A)} \models d\to c\righttherefore b\to a.
	\end{align*}
\end{enumerate} We will always write $ \mathfrak A$ instead of $\mathfrak{(A,A)}$.
\end{definition}

\todo[inline]{Example}

Computing all justifications of an arrow proportion is difficult in general, which fortunately can be omitted in many cases:

\begin{definition} We call a set $J$ of justifications a \textit{\textbf{characteristic set of justifications}} of $a\to b\righttherefore c\to d$ in $\mathfrak{(A,B)}$ iff $J$ is a sufficient set of justifications of $a\to b\righttherefore c\to d$ in $\mathfrak{(A,B)}$, that is, iff
\begin{enumerate}
	\item $J\subseteq\ \uparrow_{\mathfrak{(A,B)}}(a\to b\righttherefore c\to d)$, and
	\item $J\subseteq\ \uparrow_{\mathfrak{(A,B)}}(a\to b\righttherefore c\to d')$ implies $d'=d$, for each $d'\in\mathfrak B$.
\end{enumerate} In case $J=\{\alpha\}$ is a singleton set satisfying both conditions, we call $\alpha$ a \textit{\textbf{characteristic justification}} of $a\to b\righttherefore c\to d$ in $\mathfrak{(A,B)}$. 
\end{definition}

\todo[inline]{Example}

\section{Properties}\label{§:Properties}

In the tradition of the ancient Greeks, Lepage \cite{Lepage03} introduced (in the linguistic context) a set of properties as a guideline for formal models of analogical proportions, and his list has since been extended by a number of authors now including the following properties:\footnote{\cite{Lepage03} uses different names for his postulates --- we have decided to remain consistent with the nomenclature in \cite[§4.2]{Antic22}.} 
\begin{align*}
    \mathfrak A &\models a:b:: a:b \quad\text{(p-reflexivity)},\\
    \mathfrak{(A,B)} &\models a:b::c:d \quad\Leftrightarrow\quad (\mathfrak{B,A}) \models c:d::a:b\quad\text{(p-symmetry)},\\
    \mathfrak{(A,B)} &\models a:b::c:d \quad\Leftrightarrow\quad \mathfrak{(A,B)} \models b:a::d:c\quad\text{(inner p-symmetry)},\\
    \mathfrak A &\models a:a::a:d \quad\Leftrightarrow\quad d = a\quad\text{(p-determinism)},\\
    \mathfrak{(A,B)} &\models a:a::c:c \quad\text{(inner p-reflexivity)},\\
    \mathfrak A &\models a:b::c:d \quad\Leftrightarrow\quad \mathfrak A \models a:c::b:d \quad\text{(central permutation)},\\
    \mathfrak A &\models a:a::c:d \quad\Rightarrow\quad d = c \quad\text{(strong inner p-reflexivity)},\\
    \mathfrak A &\models a:b::a:d \quad\Rightarrow\quad d=b \quad\text{(strong p-reflexivity)}.
\end{align*} Moreover, the following property is considered, for $a,b\in A\cap B$:
\begin{align*}
    \mathfrak{(A,B)} \models a:b::b:a\quad\text{(p-commutativity).}
\end{align*} 

Furthermore, the following properties are considered, for $L$-algebras $\mathfrak{A,B,C}$ and elements $a,b\in A$, $c,d\in B$, $e,f\in C$:
\begin{prooftree}
    \AxiomC{$\mathfrak{(A,B)} \models a:b::c:d$}
        \AxiomC{$\mathfrak{(B,C)} \models c:d::e:f$}
        \RightLabel{(p-transitivity),}
    \BinaryInfC{$\mathfrak{(A,C)} \models a:b::e:f$}
\end{prooftree} and, for elements $a,b,e\in A$ and $c,d,f\in B$, the property
\begin{prooftree}
    \AxiomC{$\mathfrak{(A,B)} \models a:b::c:d$}
    \AxiomC{$\mathfrak{(A,B)} \models b:e::d:f$}
    \RightLabel{(inner p-transitivity),}
    \BinaryInfC{$\mathfrak{(A,B)} \models a:e::c:f$}
\end{prooftree} and, for elements $a\in A$, $b\in A\cap B$, $c\in B\cap C$, and $d\in C$, the property
\begin{prooftree}
    \AxiomC{$\mathfrak{(A,B)} \models a:b::b:c$}
    \AxiomC{$\mathfrak{(B,C)} \models b:c::c:d$}
    \RightLabel{(central p-transitivity).}
    \BinaryInfC{$\mathfrak{(A,C)} \models a:b::c:d$}
\end{prooftree} Notice that central p-transitivity follows from p-transitivity. 


The following theorem shows that the first-order logical framework of this paper has the same properties as the original universal algebraic framework (cf. \cite[Theorem 28]{Antic22}):

\begin{theorem}\label{t:properties} The analogical proportion relation as defined in \prettyref{d:abcd} satisfies
\begin{itemize}
    \item p-symmetry,
    \item inner p-symmetry,
    \item inner p-reflexivity,
    \item p-reflexivity,
    \item p-determinism,
\end{itemize}  and, in general, it does not satisfy
\begin{itemize}
    \item central permutation,
    \item strong inner p-reflexivity,
    \item strong p-reflexivity,
    \item p-commutativity,
    \item p-transitivity,
    \item inner p-transitivity,
    \item central p-transitivity.
\end{itemize}
\end{theorem}
\begin{proof}

We have the following proofs (some of which are very similar to the original proofs of Theorem 28 in \cite{Antic22} and are given here for completeness):
\begin{itemize}
	\item p-Symmetry and inner p-symmetry hold trivially as the framework is designed to satisfy these properties.

	\item Inner p-reflexivity follows from the fact that $x=y$ is a characteristic justification of $a\to a\righttherefore c\to c$ and $c\to c\righttherefore a\to a$.

	\item Next, we prove p-reflexivity. We first show
	\begin{align}\label{eq:abab} 
		\mathfrak A \models a\to b\righttherefore a\to b.
	\end{align} If
	\begin{align*} 
		\uparrow_\mathfrak A(a\to b)\ \cup \uparrow_\mathfrak A(a\to b)=\ \uparrow_\mathfrak A(a\to b)
	\end{align*} consists only of trivial justifications, we are done. Otherwise, there is at least one non-trivial justification in $\uparrow_\mathfrak A(a\to b)=\ \uparrow_\mathfrak A(a\to b\righttherefore a\to b)$. We proceed by showing that $\uparrow_\mathfrak A(a\to b\righttherefore a\to b)$ is $b$-maximal. For any $d\in A$, we have
	\begin{align*} 
		\uparrow_\mathfrak A(a\to b\righttherefore a\to d)\subseteq\ \uparrow_\mathfrak A(a\to b)=\ \uparrow_\mathfrak A(a\to b\righttherefore a\to b),
	\end{align*} which shows that $\uparrow_\mathfrak A(a\to b\righttherefore a\to b)$ is indeed maximal. Hence, we have shown \prettyref{eq:abab}. The same line of reasoning proves the remaining arrow proportions thus showing
	\begin{align*} 
		\mathfrak A \models a:b::a:b.
	\end{align*}

	\item Next, we prove p-determinism. ($\Leftarrow$) Inner p-reflexivity already shown above implies
	\begin{align*} 
		\mathfrak A \models a:a::a:a.
	\end{align*} $(\Rightarrow)$ We assume $\mathfrak A \models a:a::a:d$. Since $x=y\in\ \uparrow_\mathfrak A(a\to a)$, the set $\uparrow_\mathfrak A(a\to a)\ \cup \uparrow_\mathfrak A(a\to d)$ cannot consist only of trivial justifications. By definition, every justification of $a\to a\righttherefore a\to d$ is a justification of $a\to a\righttherefore a\to a$. On the other hand, we have
	\begin{align*} 
	    x = y\in\ \uparrow_\mathfrak A(a\to a\righttherefore a\to a)
	\end{align*} whereas
	\begin{align*} 
	    x = y\not\in\ \uparrow_\mathfrak A(a\to a\righttherefore a\to d),\quad\text{for all $d\neq a$}.
	\end{align*} This shows
	\begin{align*} 
	    \uparrow_\mathfrak A(a\to a\righttherefore a\to d)\subsetneq\ \uparrow_\mathfrak A(a\to a\righttherefore a\to a),
	\end{align*} which implies
	\begin{align*} 
	    \mathfrak A \models a:a\not::a:d,\quad\text{for all $d\neq a$.}
	\end{align*}

	In total, we have thus shown
	\begin{align*} 
		\mathfrak A \models a:a::a:d \quad\Leftrightarrow\quad d=a.
	\end{align*}

	\item Strong inner p-reflexivity fails for example in the structure $\mathfrak A:=(\{a,c,d\},f)$, where $f$ is a unary function, given by
	
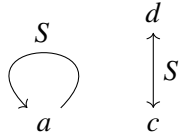
\begin{center}
	\begin{tikzpicture} 
	\node (a)               {$a$};
	\node (c) [right=of a]  {$c$};
	\node (d) [above=of c]  {$d$};
	\draw[->] (a) to [edge label'={$f$}] [loop] (a);
	\draw[<->] (c) to [edge label'={$f$}] (d);
	\end{tikzpicture}
	\end{center} As $S$ is injective, we have $\mathfrak A \models a:f(a)::c:f(d)$ which is equivalent to $\mathfrak A \models a:a::c:d$.

	\item Central permutation fails as a direct consequence of the forthcoming \prettyref{t:(A)} (depending only on inner p-reflexivity already shown above), which yields
	\begin{align*} 
	    (\{a,b,c\}) \models a:b::a:c \quad\text{whereas}\quad (\{a,b,c\}) \not\models a:a::b:c.
	\end{align*}

	Another disproof is given by
	
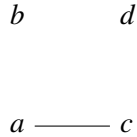
\begin{center}
	\begin{tikzpicture} 
	    \node (a)               {$a$};
	    \node (b) [above=of a]  {$b$};
	    \node (c) [right=of a]  {$c$};
	    \node (d) [right=of b]  {$d$};
	    \draw (a) to (c);
	\end{tikzpicture}
	\end{center} where $a:b::c:d$ holds trivially whereas $a:c\not::b:d$.

	\item Next, we disprove strong p-reflexivity. By the forthcoming \prettyref{t:(A)} (which depends only on inner p-reflexivity already proved above), we have
	\begin{align*} 
	    (\{a,b,d\}) \models a:b::a:d.
	\end{align*}

	\item p-Commutativity fails in the structure $\mathfrak A:=(\{a,b\},f)$, where $f$ is a unary function given by
	
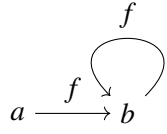
\begin{center}
	\begin{tikzpicture} 
	\node (a)               {$a$};
	\node (b) [right=of a]  {$b$};
	\draw[->] (a) to [edge label={$f$}] (b);
	\draw[->] (b) to [edge label'={$f$}] [loop] (b);
	\end{tikzpicture}
	\end{center}

	\item p-Transitivity fails, for example, in the structure $\mathfrak A:=(\{a,b,c,d,e,f\},g,h)$, where $g,h$ are unary functions given by (we omit the loops $g(o):=o$ for $o\in\{b,d,e,f\}$, and $h(o):=o$ for $o\in\{a,b,d,f\}$, in the figure)
	
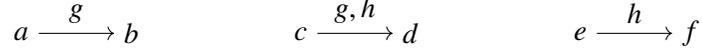
\begin{center}
	\begin{tikzpicture} 
	\node (a)               {$a$};
	\node (b) [right=of a]  {$b$};
	\node (c) [right=of b, xshift=0.8cm]  {$c$};
	\node (d) [right=of c]  {$d$};
	\node (e) [right=of d, xshift=0.8cm]  {$e$};
	\node (f) [right=of e]  {$f$};
	\draw[->] (a) to [edge label={$g$}] (b);
	\draw[->] (c) to [edge label={$g,h$}] (d);
	\draw[->] (e) to [edge label={$h$}] (f);
	\end{tikzpicture}
	\end{center}
	\todo[inline]{korrigiere wie im AP-paper (siehe Books-Version)}

	\item Inner p-transitivity fails in the structure $\mathfrak A:=(\{a,b,c,d,e,f\},g)$, where $g$ is a unary function given by (we omit the loops $g(o):=o$, for $o\in\{b,e,c,d,f\}$, in the figure)
	
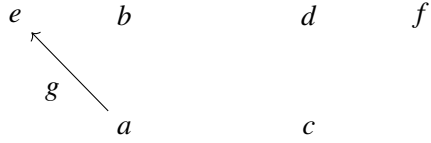
\begin{center}
	\begin{tikzpicture} 
	\node (a) {$a$};
	\node (b) [above=of a] {$b$};
	\node (e) [left=of b] {$e$};
	\node (c) [right=of a,xshift=1cm] {$c$};
	\node (d) [above=of c] {$d$};
	\node (f) [right=of d] {$f$};
	\draw[->] (a) to [edge label={$g$}] (e);
	\end{tikzpicture}
	\end{center}

	\item Central p-transitivity fails in the structure $\mathfrak A:=(\{a,b,c,d\},g,h)$, where $g,h$ are unary functions given by (we omit the loops $g(o):=o$ for $o\in\{c,d\}$, and $h(o):=o$ for $o\in\{a,d\}$, in the figure)
	
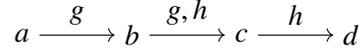
\begin{center}
	\begin{tikzpicture} 
	\node (a)               {$a$};
	\node (b) [right=of a]  {$b$};
	\node (c) [right=of b]  {$c$};
	\node (d) [right=of c]  {$d$};

	\draw[->] (a) to [edge label={$g$}] (b);
	\draw[->] (b) to [edge label={$g,h$}] (c);
	\draw[->] (c) to [edge label={$h$}] (d);
	\end{tikzpicture}
	\end{center} 


\end{itemize}
\end{proof}


The next result gives a simple characterization of the analogical proportion relation in structures consisting only of a universe:

\begin{theorem}\label{t:(A)} For any set $A$ and any $a,b,c,d\in A$, we have
\begin{align*} 
	(A) \models a:b::c:d \quad\Leftrightarrow\quad (a=b \quad\text{and}\quad c=d)\quad\text{or}\quad (a\neq b \quad\text{and}\quad c\neq d).
\end{align*}
\end{theorem}
\begin{proof} We only need to replace $z\to z$ by $x=y$ in the proof of Theorem 33 in \cite{Antic22} and we repeat the proof here for completeness.

$(\Leftarrow)$ (i) If $a=b$ and $c=d$, then $(A) \models a:b::c:d$ holds by inner p-reflexivity (\prettyref{t:properties}). (ii) If $a\neq b$ and $c\neq d$, then
\begin{align*} 
    \uparrow_{(A)}(a\to b)\ \cup \uparrow_{(A)}(c\to d)=\ \uparrow_{(A)}(b\to a)\ \cup \uparrow_{(A)}(d\to c)=\emptyset,
\end{align*} which entails $(A) \models a:b::c:d$.

$(\Rightarrow)$ By assumption, we have $(A) \models a\to b\righttherefore c\to d$. We distinguish two cases: (i) if $\uparrow_{(A)}(a\to b)\ \cup \uparrow_{(A)}(c\to d)$ consists only of trivial justifications, then we must have $a\neq b$ and $c\neq d$ since otherwise the non-trivial justification $x=y$ would be included; (ii) otherwise, $\uparrow_{(A)}(a\to b\righttherefore c\to d)$ contains the only available non-trivial justification $x=y$, which implies $a=b$ and $c=d$.
\end{proof}

\begin{corollary} In addition to the positive properties of \prettyref{t:properties}, every structure $\mathfrak A:=(A)$, consisting only of its universe, satisfies p-commutativity, inner p-transitivity, p-transitivity, central p-transitivity, and strong inner p-reflexivity.
\end{corollary}

\section{Isomorphism theorems}\label{§:ITs}

It is reasonable to expect isomorphisms --- which are structure-preserving bijective mappings between structures --- to be compatible with analogical proportions, and in this section we lift the First and Second Isomorphism Theorems in \cite{Antic22} to the setting of this paper.

\begin{lemma}[Isomorphism Lemma]\label{l:IL} For any isomorphism $ H:\mathfrak A\to\mathfrak B$ and any elements $a,b\in A$,
\begin{align*} 
	\uparrow_\mathfrak A(a\to b)=\ \uparrow_\mathfrak B(H(a)\to H(b)).
\end{align*}
\end{lemma}
\begin{proof} A direct consequence of \prettyref{l:respects}.
\end{proof}

\begin{theorem}[First Isomorphism Theorem]\label{t:FIT} For any isomorphism $H:\mathfrak A\to\mathfrak B$ and any elements $a,b\in A$, we have
\begin{align*} 
	\mathfrak{(A,B)} \models a:b::H(a):H(b).
\end{align*}
\end{theorem}
\begin{proof} Requires only minor adaptations of the proof of the First Isomorphism Theorem in \cite{Antic22}.

If $\uparrow_\mathfrak A(a\to b)\ \cup \uparrow_\mathfrak B(H(a)\to H(b))$ consists only of trivial justifications, we are done. 

Otherwise, there is at least one non-trivial justification $\alpha$ in $\uparrow_\mathfrak A(a\to b)$ or in $\uparrow_\mathfrak B(H(a)\to H(b))$, in which case the Isomorphism \prettyref{l:IL} implies that $\alpha$ is in both $\uparrow_\mathfrak A(a\to b)$ \textit{and} $\uparrow_\mathfrak B(H(a)\to H(b))$, which means that $\uparrow_{\mathfrak{(A,B)}}(a\to b\righttherefore H(a)\to H(b))$ contains at least one non-trivial justification as well. 

We proceed by showing that $\uparrow_{\mathfrak{(A,B)}}(a\to b\righttherefore H(a)\to H(b))$ is $H(b)$-maximal:
\begin{align*}
    \uparrow_{\mathfrak{(A,B)}}(a\to b\righttherefore H(a)\to H(b))
    	&=\ \uparrow_\mathfrak A(a\to b)\quad\text{(Isomorphism \prettyref{l:IL})}\\
    	&\supseteq\ \uparrow_{\mathfrak{(A,B)}}(a\to b)\ \cap \uparrow_{\mathfrak{(A,B)}}(H(a)\to d)\\
    	&=\ \uparrow_{\mathfrak{(A,B)}}(a\to b\righttherefore  H(a)\to d),
\end{align*} for every $d\in B$. An analogous argument shows the remaining directed proportions.
\end{proof}

\begin{theorem}[Second Isomorphism Theorem]\label{t:SIT} For any elements $a,b,c,d\in A$ and any isomorphism $H:\mathfrak A\to \mathfrak B$, we have
\begin{align*} 
	\mathfrak A \models a:b::c:d \quad\Leftrightarrow\quad \mathfrak B \models H(a):H(b)::H(c):H(d).
\end{align*}
\end{theorem}
\begin{proof}  An immediate consequence of the Isomorphism \prettyref{l:IL} which yields
\begin{align*} 
	&\uparrow_\mathfrak A(a\to b) = \ \uparrow_\mathfrak C(H(a)\to H(b)),\\
	&\uparrow_\mathfrak B(c\to d) = \ \uparrow_\mathfrak D(H(c)\to H(d)).
\end{align*}
\end{proof}

\begin{remark} Proportion-preserving functions have been studied in an algebraic setting in \cite{Couceiro23}.
\end{remark}

\section{Equational proportion theorem}\label{§:EPT}

In this section, we show that under certain conditions, \textit{equational dependencies} lead to an analogical proportion. That is, in some cases we expect $a:b::c:d$ to hold if $t(a,b)=t(c,d)$, for some term function $t$ --- notice that in order for the equality to make sense, $t(a,b)$ and $t(c,d)$ have to be from the same domain. The next theorem establishes a context in which this implication holds:

\begin{theorem}[Equational Proportion Theorem]\label{t:EPT} Let $a,b,c,d\in A$ and let $\mathfrak A$ be an $L$-structure.
\begin{enumerate}
\item For any c-term $t(x,y)$, if
\begin{align*} 
	t^\mathfrak A(a,b)=t^\mathfrak A(c,d),
\end{align*} and if for all $d'\neq d\in A$, we have
\begin{align*} 
	t^\mathfrak A(a,b)\neq t^\mathfrak A(c,d'),
\end{align*} then
\begin{align*} 
	\mathfrak A \models a\to b\righttherefore c\to d.
\end{align*}

\item Consequently, for any c-terms $t_a,t_b,t_c,t_d$, if
\begin{align*} 
	&t_a^\mathfrak A(a,b)=t_a^\mathfrak A(c,d) \quad\text{and}\quad t_a^\mathfrak A(a',b)\neq t_a^\mathfrak A(c,d),\quad\text{for all $a'\neq a$},\\
	&t_b^\mathfrak A(a,b)=t_b^\mathfrak A(c,d) \quad\text{and}\quad t_b^\mathfrak A(a,b)'\neq t_b^\mathfrak A(c,d),\quad\text{for all $b'\neq b$},\\
	&t_c^\mathfrak A(a,b)=t_c^\mathfrak A(c,d) \quad\text{and}\quad t_c^\mathfrak A(a,b)\neq t_c^\mathfrak A(c',d),\quad\text{for all $c'\neq c$},\\
	&t_d^\mathfrak A(a,b)=t_d^\mathfrak A(c,d) \quad\text{and}\quad t_d^\mathfrak A(a,b)\neq t_d^\mathfrak A(c,d)',\quad\text{for all $d'\neq d$},
\end{align*} then
\begin{align*} 
	\mathfrak A \models a:b::c:d.
\end{align*}
\end{enumerate}
\end{theorem}
\begin{proof} Since $t$ is a c-term and thus contains both variables $x$ and $y$, the formula
\begin{align*} 
	\alpha(x,y):\equiv (t(x,y)=t^ \mathfrak A(a,b))
\end{align*} is a c-formula. By assumption, we know that
\begin{align*} 
	t^\mathfrak A(a,b)=t^\mathfrak A(c,d),
\end{align*} which shows
\begin{align*} 
	\mathfrak A\models \alpha(a,b) \quad\text{and}\quad \mathfrak A\models \alpha(c,d).
\end{align*} This shows that $\alpha$ is a justification of $a\to b\righttherefore c\to d$ in $\mathfrak A$. It remains to show that it is a characteristic justification. For any $d'\neq d\in A$, by assumption we have
\begin{align*} 
	t^\mathfrak A(a,b)\neq t^\mathfrak A(c,d'),
\end{align*} which shows
\begin{align*} 
	\mathfrak A\not\models \alpha(c,d').
\end{align*} That is, $\alpha$ is \textit{not} a justification of $a\to b\righttherefore c\to d'$, for any $d'\neq d$. Thus, $\alpha$ is indeed a characteristic justification of $\mathfrak A \models a\to b\righttherefore c\to d$.

The second Item is an immediate consequence of the first.
\end{proof}

\begin{corollary} For any words $\mathbf{a,b,c,d}\in A^\ast$ over some alphabet $A$,\footnote{As usual, $\mathbf{ab}$ stands for the concatenation of the words $\textbf{a}$ and $\mathbf b$ and $A^\ast$ denotes the set of all words over $A$ including the empty word.}
\begin{align*} 
	\mathbf{ab = cd} \quad\Rightarrow\quad (A^\ast,\cdot) \models \textbf{a}: \mathbf b:: \mathbf c: \mathbf d.
\end{align*}
\end{corollary}

\begin{corollary} For any integers $a,b,c,d\in\mathbb Z$,
\begin{align*} 
	a+b = c+d \quad\Rightarrow\quad (\mathbb Z,+) \models a:b::c:d.
\end{align*}
\end{corollary}


\section{Equational fragment}\label{§:EF}

In many instances, it makes sense to study a restricted fragment of the full framework by syntactically restricting justifications. In this section, we look at the equational fragment consisting only of justifications of the simple form $s=t$, for some terms $s$ and $t$, and reprove the Difference Proportion Theorem in \cite{Antic22-2} from that perspective (\prettyref{t:DPT}). This provides further conceptual evidence for the robustness of the underlying framework.

\begin{definition} Let $s$ and $t$ be $L$-terms in two variables $x$ and $y$.\todo{Wo muessen $x$ und $y$ vorkommen?} Define the \textit{\textbf{set of equational justifications}} (or \textit{\textbf{e-justifications}}) of an arrow $a\to b$ in $\mathfrak A$ by
\begin{align*} 
	\uparrow^e_ \mathfrak A(a\to b):=\left\{s(x,y)=t(x,y)\in c\text-Fm_L \;\middle|\; \mathfrak A\models s(a,b)=t(a,b)\right\}.
\end{align*} extended to an arrow proportion $a\to b\righttherefore c\to d$ in a pair of $L$-structures $\mathfrak{(A,B)}$ by
\begin{align*} 
	\uparrow^e_ \mathfrak{(A,B)}(a\to b\righttherefore c\to d) :=\ \uparrow^e_ \mathfrak A(a\to b)\ \cap \uparrow^e_ \mathfrak B(c\to d).
\end{align*} The \textit{\textbf{analogical proportion relation}} $::_e$ is defined as $::$ in \prettyref{d:abcd} with $\uparrow$ replaced by $\uparrow^e$ and with the notion of triviality adapted accordingly.
\end{definition}

\begin{theorem}[Difference Proportion Theorem]\label{t:DPT} For any $a,b,c,d\in\mathbb N$,
\begin{align*} 
	(\mathbb N,S) \models_e a:b::c:d \quad\Leftrightarrow\quad a-b=c-d\quad\text{\textit{\textbf{(difference proportion)}}}.
\end{align*}
\end{theorem}
\begin{proof} Let us first compute the e-justifications in $(\mathbb N,S)$:
\begin{align*} 
	\uparrow^e_{(\mathbb N,S)}(a\to b)
		&=\left\{S^k(x)=S^\ell(y) \;\middle|\; S^k(a)=S^\ell(b),\;k,\ell\geq 0\right\}=	
		\begin{cases}
			\left\{S^{b-a+m}(x)=S^m(y) \;\middle|\; m\geq 0\right\} & a\leq b,\\
			\left\{S^m(x)=S^{a-b+m}(y) \;\middle|\; m\geq 0\right\} & b<a.
		\end{cases}
\end{align*}

$(\Rightarrow)$ Every equational justification of the form $S^k(x)=S^\ell(y)$ of $a\to b\righttherefore c\to d$ in $(\mathbb N,S)$ is a characteristic justification by the following argument: for every $d'\in\mathbb N$, we have
\begin{align*} 
	S^k(x) = S^\ell(y)\in\ \uparrow^e_{(\mathbb N,S)}(a\to b\righttherefore c\to d) \quad&\Leftrightarrow\quad c+k = \ell+d,\\
	S^k(x) = S^\ell(y)\in\ \uparrow^e_{(\mathbb N,S)}(a\to b\righttherefore c\to d') \quad&\Leftrightarrow\quad c+k = \ell+d'
\end{align*} which implies $d=d'$. 

Since $\uparrow^e_{(\mathbb N,S)}(a\to b)$ is non-empty, for all $a,b\in\mathbb N$, there must be some non-trivial equational justification $\alpha$ in $\uparrow^e_{(\mathbb N,S)}(a\to b\righttherefore c\to d)$. We distinguish two cases:
\begin{enumerate}
	\item If $\alpha\equiv (S^{b-a+m}(x)=S^m(y))$, for some $m\geq 0$, we have
	\begin{align*} 
		S^{b-a+m}(c) = S^m(x) \quad\Leftrightarrow\quad c+b-a+m = d+m \quad\Leftrightarrow\quad a-b = c-d.
\end{align*}
	\item If $\alpha\equiv (S^m(x) = S^{a-b+m}(y))$, for some $m\geq 0$, we have
	\begin{align*} 
		S^m(c) = S^{a-b+m}(d) \quad\Leftrightarrow\quad c+m = d+a-b+m \quad\Leftrightarrow\quad a-b = c-d.
\end{align*}
\end{enumerate}

$(\Leftarrow)$ We distinguish two cases:
\begin{enumerate}
	\item If $a\leq b$, we must have $c\leq d$ and
	\begin{align*} 
		S^{b-a+m}(x) = S^m(y)\in\ \uparrow^e_{(\mathbb N,S)}(a\to b\righttherefore c\to d).
\end{align*} Since $S^{b-a+m}(x)=S^m(y)$ is a characteristic equational justification, we have deduced
	\begin{align*} 
		(\mathbb N,S) \models a\to b\righttherefore_e\, c\to d \quad\text{and}\quad (\mathbb N,S) \models c\to d\righttherefore_e\, a\to b.
\end{align*} Analogously, the equational justification $S^{b-a+m}(y)=S^m(x)$ characteristically justifies
\begin{align*} 
		(\mathbb N,S) \models b\to a\righttherefore_e\, d\to c \quad\text{and}\quad (\mathbb N,S) \models d\to c\righttherefore_e\, b\to a.
\end{align*} We have thus shown
	\begin{align*} 
		(\mathbb N,S) \models a:b::_e c:d.
\end{align*}
	\item The case $b<a$ and $d<c$ is analogous.
\end{enumerate}
\end{proof}

\section{Graphs}\label{§:Graphs}

An (\textit{\textbf{undirected}}) \textit{\textbf{graph}} is a relational structure $\mathfrak G=(V_\mathfrak G,E_\mathfrak G)$, where $V_\mathfrak G$ is a set of \textit{\textbf{vertices}} of $\mathfrak G$ and $E_\mathfrak G$ consists of one or two-element sets of vertices denoting (\textit{\textbf{undirected}}) \textit{\textbf{edges}} between vertices of $\mathfrak G$. We write $a\text{ --- }_\mathfrak G\, b$ in case there is an edge between $a$ and $b$ in $\mathfrak G$. A graph $\mathfrak F$ is a \textit{\textbf{subgraph}} of $\mathfrak G$ iff $V_\mathfrak F\subseteq V_\mathfrak G$ and $E_ \mathfrak F\subseteq E_\mathfrak G$. A \textit{\textbf{path}} in a graph is a finite or infinite sequence of edges which joins a sequence of vertices. 
We write $a\stackrel{n}{\text{ --- }}_\mathfrak G b$ iff there is an (undirected) path of length $n$ between $a$ and $b$ in $\mathfrak G$, and we write $a\stackrel{\ast}{\text{ --- }}_\mathfrak G b$ iff there is some $n\geq 0$ such that $a\stackrel{n}{\text{ --- }}_\mathfrak G b$. A graph is \textit{\textbf{connected}} iff it contains a path between any two vertices. Given a first-order formula $\alpha$, we write $\mathfrak G\models \alpha$ in case $\alpha$ holds in $\mathfrak G$.

\begin{definition} Define the \textit{\textbf{$0$-path formula}} by
\begin{align*} 
	\pi_0(x,y):\equiv (x=y),
\end{align*} the \textit{\textbf{$1$-path formula}} by
\begin{align*} 
	\pi_1(x,y):\equiv (xEy),
\end{align*} and the \textit{\textbf{$n$-path formula}}, $n\geq 2$, by
\begin{align*} 
	\pi_n(x,y):\equiv (\exists z_1,\ldots,z_{n-1})(xEz_1\land\ldots\land z_{n-1}Ey).
\end{align*} The formula speaks for itself:
\begin{align*} 
	\mathfrak G\models\pi_n(a,b) 
		\quad&\Leftrightarrow\quad  \text{there is an (undirected) path of length $n$ from $a$ to $b$ in $\mathfrak G$}\\ 
		\quad&\Leftrightarrow\quad a\stackrel{n}{\text{ --- }}_\mathfrak G b.
\end{align*} We denote the set of all $n$-path formulas by $n\text-Fm$ and define the set of all \textit{\textbf{path formulas}} by $$PFm:=\bigcup_{n\geq 0}n\text-Fm.$$
\end{definition}
 
It should be mentioned that every path formula $\pi(x,y)$ is a connected formula in the sense of \prettyref{§:AP} as it contains only conjunction, and the variables $x$ and $y$ are connected which means that there is a path between them in the dependency graph of $\pi$ having vertices $x,z_1,\ldots,z_n,y$ and an edge between any two variables $v,v'$ with $vEv'$ in $\pi$.

\begin{definition} Define the \textit{\textbf{path type}} of an arrow $a\to b$ in $\mathfrak G$ by
\begin{align*} 
	\uparrow^P_\mathfrak G(a\to b):=\{\pi_n\in PFm\mid \mathfrak G\models\pi_n(a,b)\},
\end{align*} extended to an arrow proportion $a\to b \righttherefore c\to d$ in $\mathfrak{(G,H)}$ by
\begin{align*} 
	\uparrow^P_\mathfrak{(G,H)}(a\to b \righttherefore c\to d) :=\ \uparrow^P_\mathfrak G(a\to b)\ \cap \uparrow^P_\mathfrak H(c\to d).
\end{align*} Let $::_P$ denote the analogical proportion relation which is defined as $::$ with $\uparrow$ replaced by $\uparrow^P$.
\end{definition}

Our first observation is that since all considered graphs are undirected, the definition of an analogical proportion in \prettyref{d:abcd} can be simplified as follows:

\begin{lemma}\label{l:simplification} For any $a,b\in V_\mathfrak G$ and $c,d\in V_\mathfrak H$, we have
\begin{align*} 
	\mathfrak{(G,H)} \models a:b::_P c:d \quad\Leftrightarrow\quad \mathfrak{(G,H)} \models a\to b\righttherefore_P\, c\to d \quad\text{and}\quad \mathfrak{(H,G)} \models c\to d\righttherefore_P\, a\to b.
\end{align*}
\end{lemma}
\begin{proof} Since the graphs $\mathfrak G$ and $\mathfrak H$ are undirected, we have the symmetry
\begin{align*} 
	\mathfrak G\models_P\pi_n(a,b) \quad\Leftrightarrow\quad \mathfrak G\models_P\pi_n(b,a),
\end{align*} which implies
\begin{align*} 
	\uparrow^P_\mathfrak{(G,H)}(a\to b \righttherefore c\to d)=\ \uparrow^P_\mathfrak{(G,H)}(b\to a \righttherefore d\to c),
\end{align*} and which further implies
\begin{align}\label{eq:GH1} 
	\mathfrak{(G,H)} \models a\to b \righttherefore_P\, c\to d \quad\Leftrightarrow\quad \mathfrak{(G,H)} \models b\to a \righttherefore_P\, d\to c.
\end{align} Analogously, we have
\begin{align}\label{eq:GH2}
	\mathfrak{(H,G)} \models c\to d \righttherefore_P\, a\to b \quad\Leftrightarrow\quad \mathfrak{(H,G)} \models d\to c \righttherefore_P\, b\to a.
\end{align}
\end{proof}

The symmetries in \prettyref{eq:GH1} and \prettyref{eq:GH2} show that we can simplify the notation by writing
\begin{align*} 
	a\text{ --- }b \righttherefore c\text{ --- }d \quad\text{instead of}\quad a\to b\righttherefore c\to d
\end{align*} and
\begin{align*} 
	\uparrow^P_\mathfrak G(a\text{ --- }b) \quad\text{instead of}\quad \uparrow^P_\mathfrak G(a\to b) \quad\text{and}\quad \uparrow^P_\mathfrak G(b\to a)
\end{align*} and
\begin{align*} 
	\uparrow^P_\mathfrak G(a\text{ --- }b \righttherefore c\text{ --- }d) \quad\text{instead of}\quad &\uparrow^P_\mathfrak G(a\to b \righttherefore c\to d) \quad\text{and}\quad \uparrow^P_\mathfrak G(b\to a \righttherefore d\to c).
\end{align*} We shall thus rewrite the equivalence in \prettyref{l:simplification} as
\begin{align*} 
	\mathfrak{(G,H)} \models a:b::_P c:d \quad\Leftrightarrow\quad \mathfrak{(G,H)} \models a\text{ --- }b\righttherefore_P\, c\text{ --- }d \quad\text{and}\quad \mathfrak{(H,G)} \models c\text{ --- }d\righttherefore_P\, a\text{ --- }b.
\end{align*}

Notice that the path type of any edge $a\text{ --- } b$ in $\mathfrak G$ can be identified with
\begin{align*} 
	\uparrow^P_\mathfrak G(a\text{ --- }b)=\left\{n\in \mathbb N \;\middle|\; a\stackrel{n}{\text{ --- }}_\mathfrak G b\right\},
\end{align*} extended to arrow proportions by
\begin{align*} 
	\uparrow^P_\mathfrak{(G,H)}(a\text{ --- }b \righttherefore c\text{ --- }d)=\left\{n\in \mathbb N\;\middle|\; a\stackrel{n}{\text{ --- }}_\mathfrak G b, c\stackrel{n}{\text{ --- }}_\mathfrak H d\right\}.
\end{align*} This yields the following simple characterization of the analogical proportion entailment relation:

\begin{proposition}\label{p:connectivity} For any graphs $\mathfrak{G,H}$ and vertices $a,b\in V_\mathfrak G$ and $c,d\in V_\mathfrak H$, we have
\begin{align*} 
	\mathfrak{(G,H)} \models a\text{ --- }b\righttherefore_P\, c\text{ --- }d
\end{align*} iff one of the following holds:
\begin{enumerate}
	\item There is neither a path between $a$ and $b$ in $\mathfrak G$ nor between $c$ and $d$ in $\mathfrak H$; or
	\item $a\stackrel{\ast}{\text{ --- }}_\mathfrak G b$ and $c\stackrel{\ast}{\text{ --- }}_\mathfrak H d$ and there is \textit{no} $d'\neq d\in V_\mathfrak H$ such that
	\begin{enumerate}
		\item $a\stackrel{n}{\text{ --- }}_\mathfrak G b$ and $c\stackrel{n}{\text{ --- }}_\mathfrak H d$ implies $c\stackrel{n}{\text{ --- }}_\mathfrak H d'$, for all $n\geq 1$; and
		\item there is some $n\geq 1$ such that $a\stackrel{n}{\text{ --- }}_\mathfrak G b$ and $c\stackrel{n}{\text{ --- }}_\mathfrak H d'$ whereas $c\stackrel{n}{\text{ --- }}_\mathfrak H d$ does \textit{not} hold.
	\end{enumerate}
\end{enumerate} Consequently, if neither $a$ and $b$ are connected in $\mathfrak G$ nor $c$ and $d$ in $\mathfrak H$, then $\mathfrak{(G,H)} \models a:b::_P c:d$.
\end{proposition}

\begin{theorem}\label{t:properties_graphs} The analogical proportion relation in undirected graphs via path justifications satisfies
\begin{itemize}
    \item p-symmetry,
    \item inner p-symmetry,
    \item inner p-reflexivity,
    \item p-reflexivity,
    \item p-determinism,
    \item p-commutativity,
\end{itemize}  and, in general, it does not satisfy
\begin{itemize}
    \item central permutation,
    \item strong inner p-reflexivity,
    \item strong p-reflexivity,
    \item p-transitivity,
    \item inner p-transitivity,
    \item central p-transitivity,
    \item p-monotonicity.
\end{itemize}
\end{theorem}
\begin{proof} We have the following proofs:
\begin{itemize}
	\item Inner p-reflexivity follows from the fact that the $0$-path justification $\pi_0\equiv (x=y)$ is included in the path type of $a\text{ --- }b \righttherefore c\text{ --- }d$ iff $a=b$ and $c=d$, which means that it is a characteristic justification of $a\text{ --- }a \righttherefore c\text{ --- }c$, and similarly for $c\text{ --- }c \righttherefore a\text{ --- }a$.

	\item Next, we prove p-determinism. $(\Leftarrow)$ Inner p-reflexivity implies
	\begin{align*} 
		\mathfrak G \models a:a::_P a:a.
	\end{align*} $(\Rightarrow)$ An immediate consequence of the fact that $\pi_0\equiv (x=y)$ is a justification of $a\text{ --- }a \righttherefore_P a\text{ --- }a$ but not of $a\text{ --- }a \righttherefore_P a\text{ --- }d$ and the fact that every justification of the latter is trivially a justification of the former.

	\item p-Commutativity is an immediate consequence of
	\begin{align*} 
		\uparrow^P(a\text{ --- }b \righttherefore b\text{ --- }a)=\{n\in \mathbb N\mid a\stackrel{n}{\text{ --- }}b\}=\ \uparrow^P(b\text{ --- }a \righttherefore a\text{ --- }b).
	\end{align*}

	\item Central permutation fails for example in
	\begin{center}
        \begin{tikzpicture} 
        \node (a)               {$a$};
        \node (b) [above=of a]  {$b$};
        \node (c) [right=of a]  {$c$.};
        \node (d) [right=of b]  {$d$};
        \draw (a) to (c);
    \end{tikzpicture}
    \end{center} More precisely, we have $a:b::_P c:d$ by \prettyref{p:connectivity}, whereas $a:c\not ::_P b:d$ since
    \begin{align*} 
		\uparrow^P(a\text{ --- }c)\ \cup \uparrow^P(b\text{ --- }d)\neq \emptyset
	\end{align*} whereas
	\begin{align*} 
		\uparrow^P(a\text{ --- }c \righttherefore b\text{ --- }d)=\emptyset.
	\end{align*}

	\item Strong inner p-reflexivity fails for example in
	\begin{center}
    \begin{tikzpicture} 
        \node (a)               {$a$};
        \node (c) [right=of a]  {$c$};
        \node (d) [above=of c]  {$d$};
        \draw (a) to [loop] (a);
        \draw (c) to (d);
    \end{tikzpicture}
    \end{center} as we clearly have $a:a::_P c:d$ and $c\neq d$. 

	\item Strong p-reflexivity fails for example in every graph having at least three vertices and no edges as a consequence of \prettyref{p:connectivity}.

	\item p-Transitivity fails for example in
	\begin{center}
	\begin{tikzpicture} 
		\node (a) {$a$};
		\node (b) [right=of a] {$b$};
		\node (c) [right=of b] {$c$};
		\node (c') [right=of c] {};
		\node (c'') [above=of c'] {$\ast$};
		\node (d) [right=of c'] {$d$};
		\node (e) [right=of d] {$e$};
		\node (e') [right=of e] {};
		\node (e'') [above=of e'] {$\ast$};
		\node (f) [right=of e'] {$f$};
		\draw (a) to (b);
		\draw (c) -- (c'') -- (d) -- (c);
		\draw (e) -- (e'') -- (f);
	\end{tikzpicture}
	\end{center} since we clearly have
	\begin{align*} 
	a:b::_P c:d \quad\text{and}\quad c:d::_P e:f
	\end{align*} whereas
	\begin{align*} 
		\uparrow^P(a\text{ --- }b)\ \cup \uparrow^P(e\text{ --- }f)\neq \emptyset \quad\text{and}\quad \uparrow^P(a\text{ --- }b \righttherefore e\text{ --- }f) = \emptyset
	\end{align*} shows $$a:b\not::_P e:f.$$

	\item Inner p-transitivity fails for example in the graph
	
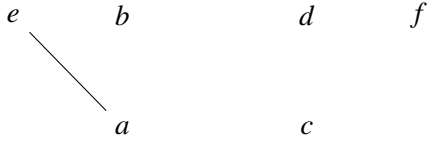
\begin{center}
    \begin{tikzpicture} 
        \node (a) {$a$};
        \node (b) [above=of a] {$b$};
        \node (e) [left=of b] {$e$};
        \node (c) [right=of a,xshift=1cm] {$c$};
        \node (d) [above=of c] {$d$};
        \node (f) [right=of d] {$f$};
        \draw (a) to (e);
    \end{tikzpicture}
    \end{center} since we clearly have
    \begin{align*} 
	a:b::_P c:d \quad\text{and}\quad b:e::_P d:f
	\end{align*} whereas 
    \begin{align*} 
		\uparrow^P(a\text{ --- }e)\ \cup \uparrow^P(c\text{ --- }f)\neq\emptyset \quad\text{and}\quad \uparrow^P(a\text{ --- }e \righttherefore c\text{ --- }f)=\emptyset
	\end{align*} shows $$a:e\not::_P c:f.$$

	\item Central p-transitivity fails for example in
	
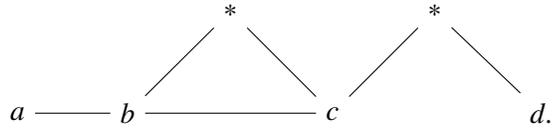
\begin{center}
	\begin{tikzpicture} 
		\node (a) {$a$};
		\node (b) [right=of a] {$b$};
		\node (b') [right=of b] {};
		\node (b'') [above=of b'] {$\ast$};
		\node (c) [right=of b'] {$c$};
		\node (c') [right=of c] {};
		\node (c'') [above=of c'] {$\ast$};
		\node (d) [right=of c'] {$d$.};
		\draw (a) -- (b);
		\draw (b) -- (b'') -- (c) -- (b);
		\draw (c) -- (c'') -- (d);
	\end{tikzpicture}
	\end{center} The proof is analogous to the disproof of p-transitivity.

	\item Finally, we disprove p-monotonicity. For this, consider the graph $\mathfrak F$
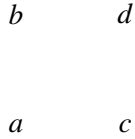
	\begin{center}
	\begin{tikzpicture} 
		\node (a)               {$a$};
	    \node (b) [above=of a]  {$b$};
	    \node (c) [right=of a]  {$c$};
	    \node (d) [above=of c]  {$d$};
	\end{tikzpicture}
	\end{center} consisting of four vertices and no edges. By \prettyref{p:connectivity}, we have
	\begin{align*} 
		\mathfrak F \models a:b::_P c:d.
	\end{align*} The graph $\mathfrak F$ is a subgraph of $\mathfrak G$ given by
	
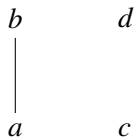
\begin{center}
    \begin{tikzpicture} 
        \node (a)               {$a$};
        \node (b) [above=of a]  {$b$};
        \node (c) [right=of a]  {$c$};
        \node (d) [right=of b]  {$d$};
        \draw[-] (a) to (b);
    \end{tikzpicture}
    \end{center} where we have
    \begin{align*} 
		\mathfrak G \not\models a:b::_P c:d.
	\end{align*}
\end{itemize} 
\end{proof}

\begin{remark}\label{rem:p-commutativity} The above validity of p-commutativity in undirected graphs with respect to path justifications is interesting as it is the first known class of structures to satisfy this property and it is the only difference to the properties of the general framework where p-commutativity fails (cf. \prettyref{t:properties}).
\end{remark}

Let $\mathfrak G_ \mathbb N$ denote the infinite undirected graph with $V_{\mathfrak G_ \mathbb N}:=\mathbb N$ which is obtained by adding an undirected edge between $a$ and $a+1$, for every $a\in\mathbb N$: 
\begin{center}
\begin{tikzpicture} 
	\node (0) {$0$};
	\node (1) [right of=0] {$1$};
	\node (2) [right of=1] {$2$};
	\node (3) [right of=2] {\ldots};
	\draw (0) to (1);
	\draw (1) to (2);
	\draw (2) to (3);
\end{tikzpicture}
\end{center}
The next result shows that we can characterize the $n$-path relation $c\stackrel{n}{\text{ --- }}_\mathfrak H d$ in the target domain $\mathfrak H$ via analogical proportions using $\mathfrak G_ \mathbb N$ as the source domain:

\begin{theorem}\label{t:G(N)_H} For any $a,b\in\mathbb N$ and $c,d\in V_\mathfrak H$,
\begin{align*} 
	(\mathfrak G_ \mathbb N, \mathfrak H) \models a:b::_P c:d \quad\Leftrightarrow\quad c\stackrel{|a-b|}{\text{---}\text{---}}_\mathfrak H d.
\end{align*} Consequently,
\begin{align*} 
	(\mathfrak G_ \mathbb N, \mathfrak H) \models 0:n::_P c:d \quad\Leftrightarrow\quad c\stackrel{n}{\text{ --- }}_\mathfrak H d.
\end{align*}
\end{theorem}
\begin{proof} Since there is exactly one path of length $|a-b|$ between any two vertices $a,b\in V_{\mathfrak G_ \mathbb N}$, we have
\begin{align*} 
	\uparrow^P_{\mathfrak G_ \mathbb N}(a\text{ --- } b) = \{n\in \mathbb N\mid n\geq |a-b| \text{ and } n\equiv |a-b|\quad (mod\;2)\}
\end{align*} which implies that
\begin{align*} 
	\uparrow^P_{(\mathfrak G_ \mathbb N, \mathfrak H)}(a\text{ --- }b \righttherefore c\text{ --- }d) = \{n\in \mathbb N\mid n\geq k \text{ and } n\equiv |a-b|\quad (mod\;2)\},
\end{align*} where $k$ is the smallest nonnegative integer such that $k\equiv |a-b|\quad (mod\;2)$ and $a\stackrel{k}{\text{---}}_ \mathfrak H b$, provided that such a number $k$ exists. If such a number does not exist, we have $\uparrow^P_{\mathfrak G_ \mathbb N}(a\text{ --- } b \righttherefore c\text{ --- } d) = \emptyset$.
\end{proof}

Interestingly enough, the next result shows that difference proportions in the structure of natural numbers (cf. \prettyref{t:DPT}) occur naturally in the graph-representation as well.

\begin{theorem}[Difference Proportion Theorem] For any $a,b,c,d\in\mathbb N$,
\begin{align*} 
	\mathfrak G_ \mathbb N \models a:b::_P c:d \quad\Leftrightarrow\quad |a-b|=|c-d|.
\end{align*} 
\end{theorem}
\begin{proof} A direct consequence of \prettyref{t:G(N)_H}.
\end{proof}



\section{Conclusion}

The purpose of this paper was to lift the abstract algebraic framework of analogical proportions in \cite{Antic22} from universal algebra to the strictly more expressive setting of full first-order logic. This was achieved by extending abstract rewrite to connected justifications containing arbitrary quantification and relations but disallowing the use of disjunction and negation. We have shown that the extended framework preserves all desired properties, and we have shown the brand new Equational Proportion \prettyref{t:EPT} not provable in the purely algebraic setting. We have analyzed analogical proportions in the relational structure of graphs.

The major line of future research is to further lift the concepts and results of this paper from first-order to second-order and, ultimately, to higher-order logic containing quantified functions and relations (see e.g. \cite{Leivant94}).
This is desirable since some proportions cannot be expressed in first-order logic. For example, in the structure with two relations $P$ and $R$ given by
\begin{center}
\begin{tikzpicture} 
	\node (a) {$a$};
	\node (b) [below=of a] {$b$};
	\node (c) [right=of a] {$c$};
	\node (d) [below=of c] {$d$};
	\draw[->] (a) to [edge label'={$P$}] (b);
	\draw[->] (c) to [edge label={$R$}] (d);
\end{tikzpicture}
\end{center} the set of justifications of $a\to b \righttherefore c\to d$ is empty, whereas in second-order logic it contains the justification $(\exists S)S(x,y)$. That is, second-order and higher-order logic allow us to detect similarities which remain undetected in first-order logic.

\bibliographystyle{acm}
\bibliography{/Users/christianantic/Bibdesk/Bibliography,/Users/christianantic/Bibdesk/Publications_J,/Users/christianantic/Bibdesk/Publications_C,/Users/christianantic/Bibdesk/Preprints}

\begin{thebibliography}{1}

\bibitem{Antic22}
{\sc Anti\'c, C.}
\newblock Analogical proportions.
\newblock {\em Annals of Mathematics and Artificial Intelligence 90}, 6 (2022),
  595--644.
\newblock \url{https://doi.org/10.1007/s10472-022-09798-y}.

\bibitem{Antic22-2}
{\sc Anti\'c, C.}
\newblock Analogical proportions in monounary algebras.
\newblock {\em Annals of Mathematics and Artificial Intelligence 92\/} (2024),
  1663--1677.
\newblock \url{https://doi.org/10.1007/s10472-023-09921-7}.

\bibitem{Antic21-3}
{\sc Anti\'c, C.}
\newblock Boolean proportions.
\newblock {\em Logical Methods in Computer Science 20}, 2 (2024), 2:1 -- 2:20.
\newblock \url{https://doi.org/10.46298/lmcs-20(2:2)2024}.

\bibitem{Antic23-23}
{\sc Anti\'c, C.}
\newblock Logic program proportions.
\newblock {\em Annals of Mathematics and Artificial Intelligence 93\/} (2025),
  321--342.
\newblock \url{https://doi.org/10.1007/s10472-023-09904-8}.

\bibitem{Couceiro23}
{\sc Couceiro, M., and Lehtonen, E.}
\newblock Galois theory for analogical classifiers.
\newblock {\em Annals of Mathematics and Artificial Intelligence 92\/} (2024),
  29--47.
\newblock \url{https://doi.org/10.1007/s10472-023-09833-6}.

\bibitem{Herzig24}
{\sc Herzig, A., Lorini, E., and Prade, H.}
\newblock A novel view of analogical proportions between formulas.
\newblock In {\em ECAI 2024}. 2024, pp.~1270--1277.

\bibitem{Hinman05}
{\sc Hinman, P.~G.}
\newblock {\em {Fundamentals of Mathematical Logic}}.
\newblock A K Peters, Wellesley, MA, 2005.

\bibitem{Leivant94}
{\sc Leivant, D.}
\newblock Higher order logic.
\newblock In {\em Handbook of Logic in Artificial Intelligence and Logic
  Programming}, vol.~2. 1994, pp.~229--322.

\bibitem{Lepage03}
{\sc Lepage, Y.}
\newblock {\em {De L'Analogie. Rendant Compte de la Commutation en
  Linguistique}}.
\newblock Ha{\-}bi{\-}li{\-}ta{\-}tion \`a diriger les recherches, Universit\'e
  Joseph Fourier, Grenoble, 2003.

\end{thebibliography}
\if\isdraft1
\newpage

\section*{Submission}

Submitted to TOCL in Oct 2023.

\vspace*{0.5cm}

Falls abgelehnt wird, reiche nur noch bei reinen Logik-Journalen ein: APAL, ... 

\section{}

\todo[inline]{$\mathbf o\in T_L(\emptyset)^{\ldots}$ statt $\in Cs_L$}

\section{}

We now want to compare the equational fragment consisting of justifications $s=t$ to the original rewrite fragment consisting of justifications of the form $s\to t$.

\begin{theorem}\label{t:incomparable} The relations $::_{ar}$ and $::_e$ are incomparable, that is, there is an $L$-algebra $\mathfrak A$ and elements $a,b,c,d\in A$ such that
\begin{align*} 
	\mathfrak A_1\models_r a:b:: c:d \quad\not\Rightarrow\quad \mathfrak A_1\models_{ae} a:b:: c:d,
\end{align*} and there is an $L$-algebra $\mathfrak A_2$ and elements $a_2,b_2,c_2,d_2\in A_2$ such that
\begin{align*} 
	\mathfrak A_2\models_{ae} a_2:b_2:: c:d \quad\not\Rightarrow\quad \mathfrak A_2\models_r a:b:: c:d.
\end{align*}
\end{theorem}
\begin{proof} First, we shall construct an algebra $\mathfrak A$ such that
\begin{align*} 
	\mathfrak A\models_{ar} a:b:: c:d \quad\text{wheres}\quad \mathfrak A\not\models_{ae} a:b:: c:d,\quad\text{for some $a,b,c,d\in A$.}
\end{align*} For this consider the algebra $\mathfrak A:=([7],f)$\footnote{Define $[n]:=\{1,\ldots,n\}$, for all $n\in\mathbb N$.} given by
\begin{center}
\begin{tikzpicture} 
\node (1) {$1$};
\node (2) [right=of 1] {$2$};
\node (3) [right=of 2] {$3$};
\node (4) [right=of 3] {$4$};
\node (5) [right=of 4] {$5$};
\node (6) [right=of 5] {$6$};
\node (7) [right=of 6] {$7$};
\draw[->] (1) to [edge label'={$f$}] (2);
\draw[->] (2) to [edge label'={$f$}][loop] (2);
\draw[->] (3) to [edge label={$f$}] (2);
\draw[->] (4) to [edge label'={$f$}] (5);
\draw[->] (5) to [edge label'={$f$}][loop] (5);
\draw[->] (6) to [edge label={$f$}] (5);
\draw[->] (7) to [edge label'={$f$}][loop] (7);
\end{tikzpicture}
\end{center} We have
\begin{align*} 
	ar\text-Type_\mathfrak A(1\to 3)\cup ar\text-Type_\mathfrak A(4\to 7)=ar\text-Type_\mathfrak A(3\to 1)\cup Jus_\mathfrak A(7\to 4)=\emptyset,
\end{align*} which shows
\begin{align*} 
	\mathfrak A\models_{ar} 1 :3 :: 4 :7.
\end{align*} On the other hand, we have
\begin{align*} 
	\uparrow^e_\mathfrak A(1\to 3)\ \cup \uparrow^e_\mathfrak A(4\to 7)=\{f(y_1)=f(y_2),\ldots\}
\end{align*} whereas
\begin{align*} 
	\uparrow^e_\mathfrak A(1\to 3 \righttherefore 4\to 7)=\emptyset,
\end{align*} which shows
\begin{align*} 
	\mathfrak A\not\models_{ae} 1 :3 :: 4 :7.
\end{align*}

Second, we shall construct an algebra $\mathfrak A$ such 
\begin{align*} 
	\mathfrak A\models_{ae} a:b:: c:d \quad\text{whereas}\quad \mathfrak A\not\models_{ar} a:b:: c:d,\quad\text{for some $a,b,c,d\in A$.}
\end{align*} For this consider the algebra...
\todo[inline]{adapt to ``abstract''}
\end{proof}

\section{}

In \cite[§7.3]{Antic22}, we argued that the set of all justifications of all pairs of elements does \textit{not} form a category by giving the following counterexample: let $\mathfrak A:=([1,3],S)$ such that (we omit the loop $S(2):=2$ in the figure)
\begin{center}
\begin{tikzpicture} 
	\node (a) {$1$};
	\node (b) [right=of a] {$2$};
	\node (c) [right=of b] {$3$};

	\draw[->] (a) to [edge label'={$S$}] (b);
	\draw[->] (c) to [edge label={$S$}] (b);
\end{tikzpicture}
\end{center} We have
\begin{align*} 
    z\to f(z)\in Jus(1,2) \quad\text{and}\quad f(z)\to z\in Jus(2,1),
\end{align*} but there is no non-trivial justification in $Jus(1,3)$. The situation changes, however, if we allow equational justifications as there is the equational justification $S(y_1)=S(y_2)\in\ \uparrow^e(1,3)$. This leads us to the question, whether equational justifications yield categories. This is related to the question whether (inner) p-transitivity holds with respect to atomic entailment.

\section{Connected formulas}

The following formulas are not connected according to the current definition (recall that until recently, the logic-based framework was formulated for \textit{abstract} formulas only):
\begin{align*} 
    &x=a\land y=a\\
    &x=z\land z=2\land y=2\\
    &x\leq z\land z=2\land z\leq y.
\end{align*} This indicates that we should add constant symbols as vertices to the dependency graph of a formula!

The following formulas are connected:
\begin{align*} 
    &x=y\land x\leq 2\\
    &x\neq y\land y\not\leq 2.
\end{align*}

\section{}

\todo[inline]{eventuell paper ``Proportional analogies'' hier noch einfuegen}

\fi
\end{document}